%% file: main-edcs.tex
\documentclass[11pt]{article}

\usepackage{amsthm}
\usepackage{graphicx} 
\usepackage{array} 

\usepackage{amsmath, amssymb, amsfonts, verbatim}
\usepackage{hyphenat,epsfig,multirow}

\usepackage[usenames,dvipsnames]{xcolor}
\usepackage[ruled]{algorithm2e}

\usepackage{tcolorbox}
\tcbuselibrary{skins,breakable}
\tcbset{enhanced jigsaw}

\usepackage[compact]{titlesec}

\definecolor{DarkRed}{rgb}{0.5,0.1,0.1}
\definecolor{DarkBlue}{rgb}{0.1,0.1,0.5}

\usepackage{hyperref}
\hypersetup{
colorlinks=true,
pdfnewwindow=true,
citecolor=ForestGreen,
linkcolor=DarkRed,
filecolor=DarkRed,
urlcolor=DarkBlue
}

\usepackage{tikz}
\usetikzlibrary{arrows}
\usetikzlibrary{arrows.meta}
\usetikzlibrary{shapes}
\usetikzlibrary{backgrounds}
\usetikzlibrary{positioning}
\usetikzlibrary{decorations.markings}
\usetikzlibrary{patterns}
\usetikzlibrary{calc}

\usepackage{caption}
\usepackage{subcaption}

\usepackage{bm}
\usepackage{url}
\usepackage{xspace}
\usepackage[mathscr]{euscript}

\usepackage{mdframed}

\usepackage[noend]{algpseudocode}
\makeatletter
\def\BState{\State\hskip-\ALG@thistlm}
\makeatother

\usepackage{cite}
\usepackage{enumerate}

\usepackage[margin=1in]{geometry}

\newtheorem{theorem}{Theorem}
\newtheorem{lemma}{Lemma}[section]
\newtheorem{proposition}[lemma]{Proposition}

\newtheorem{claim}[lemma]{Claim}

\newtheorem{definition}{Definition}

\newtheorem{problem}{Problem}
\newtheorem{remark}[lemma]{Remark}

\newtheorem*{claim*}{Claim}
\newtheorem*{proposition*}{Proposition}
\newtheorem*{lemma*}{Lemma}
\newtheorem*{problem*}{Problem}

\newtheorem{mdresult}{Result}
\newenvironment{result}{\begin{mdframed}[backgroundcolor=lightgray!40,topline=false,rightline=false,leftline=false,bottomline=false,innertopmargin=2pt]\begin{mdresult}}{\end{mdresult}\end{mdframed}}

\newtheorem{mdinvariant}{Invariant}

{\hspace*{\fill}$\Box$\par}

\allowdisplaybreaks

\renewcommand{\qed}{\nobreak \ifvmode \relax \else
      \ifdim\lastskip<1.5em \hskip-\lastskip
      \hskip1.5em plus0em minus0.5em \fi \nobreak
      \vrule height0.75em width0.5em depth0.25em\fi}

\input{macros}

\title{Towards a Unified Theory of Sparsification for Matching Problems}
\author{Sepehr Assadi\footnote{\texttt{sassadi@cis.upenn.edu}. Supported in part by the National Science Foundation grant CCF-1617851.}  \\ University of Pennsylvania
\and Aaron Bernstein\footnote{\texttt{bernstei@gmail.com}.} \\ Rutgers University
}

\date{}

\begin{document}
\maketitle

\begin{abstract}
\input{abstract}
\end{abstract}

\input{intro}

\input{prelim}

\input{edcs}

\input{communication}

\input{stochastic}

\input{fault-tolerant}

\subsection*{Acknowledgements}
Sepehr Assadi is grateful to his advisor Sanjeev Khanna for many helpful discussions, and to Soheil Behnezhad for sharing a write-up of~\cite{BehnezhadFHR19}.

\bibliographystyle{abbrv}
\bibliography{general}

\appendix

\input{appendix}

\end{document}

%% file: macros.tex
\newcommand{\toShrink}{-.20cm}
\newcommand{\toShrinkEnu}{-.2cm}

\newcommand{\Qed}[1]{\hspace*{\fill} $\qed_{\textnormal{~#1}}$}

\newcommand{\Mstar}{\ensuremath{M^{\star}}}

\newcommand{\eps}{\ensuremath{\varepsilon}}
\newcommand{\Paren}[1]{\Big(#1\Big)}

\newcommand{\bracket}[1]{\left[#1\right]}
\newcommand{\paren}[1]{\ensuremath{\left(#1\right)}\xspace}
\newcommand{\card}[1]{\left\vert{#1}\right\vert}

\newcommand{\expect}[1]{\Exp\bracket{#1}}

\newcommand{\set}[1]{\ensuremath{\left\{ #1 \right\}}}

\DeclareMathOperator*{\Exp}{\ensuremath{{\mathbb{E}}}}
\DeclareMathOperator*{\Prob}{\ensuremath{\textnormal{Pr}}}
\renewcommand{\Pr}{\Prob}

\newcommand{\Ex}{\Exp}

\newcommand{\etal}{{\it et al.\,}}

\newcommand{\event}[1]{\ensuremath{{\sf E}_{#1}}}

\newcommand{\rs}{Ruzsa-Szemer\'{e}di\xspace}

\newenvironment{tbox}{\begin{tcolorbox}[
		enlarge top by=5pt,
		enlarge bottom by=5pt,
		 breakable,
		 drop shadow=black,
		 boxsep=0pt,
                  left=4pt,
                  right=4pt,
                  top=10pt,
                  arc=0pt,
                  boxrule=1pt,toprule=1pt,
                  colback=white
                  ]
	}
{\end{tcolorbox}}

\newcommand{\bA}{\bm{A}}

\newcommand{\bB}{\ensuremath{\bm{B}}}

\renewcommand{\event}{\mathcal{E}}


%% file: abstract.tex
	In this paper, we present a construction of a ``matching sparsifier'', that is, a sparse subgraph of the given graph that preserves large matchings approximately and is robust to modifications of the graph. 
	We use this matching sparsifier to obtain several new algorithmic results for the maximum matching problem: 
	
	\begin{itemize}
		\item An almost $(3/2)$-approximation one-way communication protocol for the maximum matching problem, significantly simplifying 
		the $(3/2)$-approximation protocol of Goel, Kapralov, and Khanna (SODA 2012) and extending it from bipartite graphs to general graphs. 
		\item An almost $(3/2)$-approximation algorithm for the stochastic matching problem, improving upon and significantly simplifying the previous $1.999$-approximation algorithm of Assadi, Khanna, and Li (EC 2017). 
		\item An almost $(3/2)$-approximation algorithm for the fault-tolerant matching problem, which, to our knowledge, is the first non-trivial algorithm for this problem.  
	\end{itemize}
	
	\medskip
	
	Our matching sparsifier is obtained by proving new properties of the edge-degree constrained subgraph (EDCS) of Bernstein and Stein (ICALP 2015; SODA 2016)---designed in the context of maintaining matchings in dynamic graphs---that 
	identifies EDCS as an excellent choice for a matching sparsifier. This leads to surprisingly simple and non-technical proofs of the above results in a unified way. 
	Along the way, we also provide a much simpler proof of the fact that an EDCS is guaranteed to contain a large matching, which may be of independent interest.
	

%% file: intro.tex
\renewcommand{\deg}[2]{\ensuremath{\textnormal{deg}_{#1}(#2)}}

\newcommand{\degG}[1]{\deg{}{#1}}

\newcommand{\MM}[1]{\ensuremath{\mu(#1)}}

\newcommand{\tG}{\widetilde{G}}
\newcommand{\tH}{\widetilde{H}}
\newcommand{\tE}{\widetilde{E}}
\newcommand{\tbeta}{\widetilde{\beta}}

\section{Introduction}\label{sec:intro}

	A common tool for dealing with massive graphs is sparsification. Roughly speaking, a sparsifier of a graph $G$ is a subgraph $H$ that (approximately) preserves certain properties of $G$ while having a smaller number of 
	edges. Such sparsifiers have been studied in great detail for various properties: for example, a spanner~\cite{Awerbuch85,PelegS89} or a distance preserver~\cite{BollobasCE03,CoppersmithE06} preserves pairwise distances,
	 a cut sparsifier~\cite{Karger94,BenczurK96,FungHHP11} preserves cut information, and a spectral sparsifier~\cite{SpielmanT11,BatsonSS09} preserves spectral properties of the graph.
	An additional property that we often require of a graph sparsifier is \emph{robustness}: it should continue to be a good sparsifier even as the graph changes.
	Some sparsifiers are robust by nature (e.g cut sparsifiers), but others (e.g spanners) are not, and for this reason there is an extensive literature on designing sparsifiers that can provide additional robustness guarantees.

	In this paper, we study the problem of designing robust sparsifiers for the prominent problem of \emph{maximum matching}.
	Multiple notions of sparsification for the matching problem have already been identified in the literature.
	One example is a subgraph that preserves the
	largest matching inside any given subset of vertices in $G$ approximately. This notion is also known
	as a {matching cover} or a {matching skeleton}~\cite{GoelKK12,LeeS17} in the literature and is closely related to the communication and streaming complexity of the matching problem. Another example of a sparsifier 
	is a subgraph that can preserve the largest matching on {random} subsets of edges of $G$, a notion closely related to the stochastic matching problem~\cite{BlumDHPSS15,AssadiKL17ec}. 
	An example of a robust sparsifier for matching is a fault-tolerant subgraph, namely a subgraph $G$ that continue to preserve large matchings in $G$ even after a fraction of the edges is deleted by an adversary. 
	As far as we know, the fault-tolerant matching problem has not previously been studied, but it is a natural model to consider as it has received lots of attention in the context of spanners
	and distance preservers (see e.g. ~\cite{ChechikLPR09,Peleg09,BaswanaCR16,BodwinGPW17,BodwinDPW18}).

	 Our \underline{first contribution} is a subgraph $H$ that we show is a robust matching sparsifier in \emph{all} of the senses above. Our result is thus the first to unify these notions of sparsification for the maximum matching problem. In addition to 
	 unifying, our construction yields improved results for each individual notion of sparsification and the corresponding problems, namely, the one-way communication complexity of matching, stochastic matching, and fault-tolerant 
	 matching problems. Interestingly, our unified approach allows us to also provide much simpler proofs than all previously existing work for these problems. 
	 The subgraph we use as our sparsifier comes from a pair of papers by Bernstein and Stein on dynamic matching \cite{BernsteinS15,BernsteinS16}---they refer to this subgraph as an edge-degree 
	 constrained subgraph (EDCS for short).  The EDCS was also very recently used in~\cite{AssadiBBMS17} to design sublinear algorithms for matching across several different models for massive graphs. Our applications of the EDCS in the current paper, as well as the new properties we prove for the EDCS, are quite different from those in ~\cite{BernsteinS15,BernsteinS16,AssadiBBMS17}. Our first contribution thus takes an existing subgraph, and then provides the first proofs that it satisfies the three notions of sparsification described above. 
	 
	 Our \underline{second contribution} is a {much} simpler (and even slightly improved) proof of the main property of an EDCS in previous work proved in \cite{BernsteinS15,BernsteinS16}, namely that an EDCS contains a large matching of the original 
	 graph. Our new proof significantly simplifies the analysis of \cite{BernsteinS16} and allows for simple and self-contained proofs of the results in this paper.
	 
	 \paragraph{Definition of the EDCS.}
	 Before stating our results, we give a definition of the EDCS from~\cite{BernsteinS15,BernsteinS16}, as this is the subgraph we use for all of our results (see Section~\ref{sec:prelim} for more details).

	 \begin{definition}[\!\!\cite{BernsteinS15}]\label{def:edcs}
	 	For any graph $G(V,E)$ and integers $\beta \geq \beta^- \geq 0$, an \emph{edge-degree constrained subgraph (EDCS)} $(G,\beta,\beta^-)$ is a subgraph $H:=(V,E_H)$ of $G$ with the following two properties: 
	 	\begin{enumerate}[(P1)]
	 		\item\label{long-property1} For any edge $(u,v) \in E_H$: $\deg{H}{u} + \deg{H}{v} \leq \beta$. 
	 		\item\label{long-property2} For any edge $(u,v) \in E \setminus E_H$: $\deg{H}{u} + \deg{H}{v} \geq \beta^-$.  
	 	\end{enumerate}

	 \end{definition}
 
	It is not hard to show that an EDCS of a graph $G$ always exists for any parameters $\beta > \beta^-$ and that it is sparse, i.e., only has $O(n\beta)$ edges. 
	A key property of EDCS proven previously~\cite{BernsteinS15,BernsteinS16} (and simplified in our paper) 
	is that for any reasonable setting of the parameters (e.g. $\beta^-$ being sufficiently close to $\beta$), any EDCS $H$ of $G$ contains an (almost) $3/2$ approximate  
	matching of $G$.

\subsection{Our Results and Techniques}\label{sec:results}

We now give detailed definitions of the notions of sparsification and the corresponding problems addressed in this paper, as well as our results for each one. 
Our second contribution---a significantly simpler proof that an EDCS contains an almost $(3/2)$-approximate matching---is left for Section~\ref{sec:edcs}.

\paragraph{One-Way Communication Complexity of Matching.} Consider the following two-player communication problem: Alice is given a graph $G_A(V,E_A)$ and Bob holds a graph $G_B(V,E_B)$. 
The goal for Alice is to send a single message to Bob such that Bob outputs an approximate maximum matching in $E_A \cup E_B$. What is the minimum length of the message, i.e., the one-way communication complexity, for achieving a certain fixed approximation ratio on all graphs? One can show that the message communicated by Alice to Bob is indeed a matching skeleton, namely a data structure (but not necessarily a subgraph), that allows Bob to find 
a large matching in a given subset of vertices in Alice's input (see~\cite{GoelKK12} for more details).

This problem was first studied by Goel, Kapralov, and Khanna~\cite{GoelKK12} (see also the subsequent paper of Kapralov~\cite{Kapralov13}), owing to its close connection to
one-pass streaming algorithms for matching. Goel~\etal~\cite{GoelKK12} designed an algorithm that achieves a $(3/2)$-approximation in bipartite graphs using only $O(n)$ communication and proved
that any better than $(3/2)$-approximation protocol requires $n^{1+\Omega(\frac{1}{\log\log{n}})}$ communication even on bipartite graphs (see, e.g.~\cite{GoelKK12,AssadiKL17} for further details on this lower bound). 
A follow-up work by Lee and Singla~\cite{LeeS17} further generalized the algorithm of~\cite{GoelKK12} to general graphs, albeit with a slightly worse approximation ratio of $5/3$ (compared to $3/2$ of~\cite{GoelKK12}). 

We extends the results in~\cite{GoelKK12} to general graphs with almost no loss in approximation.

\begin{result}\label{res:communication}
	For any constant $\eps > 0$, the protocol where Alice computes an EDCS of her graph with $\beta=O(1)$ and $\beta^-=\beta-1$ and sends it to Bob is a $(3/2+\eps)$-approximation one-way communication protocol for the 
	maximum matching problem with uses $O_{}(n)$ communication. 
\end{result}

We remark that both the previous algorithm of~\cite{GoelKK12} as well as its extension in~\cite{LeeS17} are quite involved and rely on a fairly complicated graph decomposition as 
well as an intricate primal-dual analysis. As such, we believe that the main contribution in Result~\ref{res:communication} is in fact in providing a simple and self-contained proof of this result. 

\paragraph{Stochastic Matching.} In the stochastic matching problem, we are given a graph $G(V,E)$ and a probability parameter $p \in (0,1)$. A realization of $G$ is a subgraph $G_p(V,E_p)$ obtained by picking 
each edge in $G$ independently with probability $p$ to include in $E_p$. The goal in this problem is to find a subgraph $H$ of $G$ with max-degree bounded by a function of $p$ (independent of number of vertices), such that
the size of maximum matching in realizations of $H$ is close to size of maximum matching in realizations of $G$. It is immediate to see that $H$ in this problem is simply a sparsifier of $G$ which preserves large matchings on random subsets of 
edges. 

This problem was first introduced by Blum~\etal~\cite{BlumDHPSS15} primarily to model the kidney exchange setting and has since been studied extensively in the literature~\cite{AssadiKL16ec,AssadiKL17ec,BehnezhadR18,YamaguchiM18}. 
Early algorithms for this problem in~\cite{BlumDHPSS15,AssadiKL16ec} (and the later ones for the weighted variant of the problem~\cite{BehnezhadR18,YamaguchiM18}) all had approximation ratio at least $2$, naturally
raising the question that whether $2$ is the best approximation ratio achievable for this problem. Assadi, Khanna, and Li~\cite{AssadiKL17ec} ruled out this perplexing possibility by obtaining a slightly better than $2$-approximation
algorithm for this problem, namely an algorithm with approximation ratio close to $1.999$ (which improves to $1.923$ for small $p$).

We prove that using an EDCS results in a significantly improved algorithm for this problem. 

\begin{result}\label{res:stochastic}
	For any constant $\eps > 0$, an $EDCS$ of $G$ with $\beta = O(\frac{\log{(1/p)}}{p})$ and $\beta^- = \beta - 1$ achieves a
	$(3/2+\eps)$-approximation algorithm for the stochastic matching problem with a subgraph of maximum degree $O(\frac{\log{(1/p)}}{p})$.
\end{result}
\noindent
We remark that our bound on the maximum degree in Result~\ref{res:stochastic} is optimal (up to an $O(\log{(1/p)})$ factor) for any constant-factor approximation algorithm (see~\cite{AssadiKL17ec}).
In addition to significantly improving upon the previous best algorithm of~\cite{AssadiKL17ec}, our Result~\ref{res:stochastic} is much simpler than that of~\cite{AssadiKL17ec}, in terms of the both the algorithm and (especially) the analysis. 

\emph{Remark.} Independently and concurrently, Behnezhad~\etal~\cite{BehnezhadFHR19} also presented an algorithm for stochastic matching with a subgraph of max-degree $O(\frac{\log{(1/p)}}{p})$
that achieves an approximation of almost $(4\sqrt{2}-5)$ ($\approx 0.6568$ compared to $0.6666$ in Result~\ref{res:stochastic}). They also provided an algorithm with approximation ratio strictly better than half 
for weighted stochastic matching (our result does not work for weighted graphs). In terms of techniques, our paper and~\cite{BehnezhadFHR19} are entirely disjoint. 

\paragraph{Fault-Tolerant Matching.} Let $f \geq 0$ be an integer, $G(V,E)$ be a graph, and $H$ be any subgraph of $G$. We say that $H$ is an $\alpha$-approximation $f$-tolerant subgraph of $G$ iff
for any subset $F \subseteq E$ of size $\leq f$, the 
maximum matching in $H \setminus F$ is an $\alpha$-approximation to maximum matching in $G \setminus F$ -- that is, $H$ is a robust sparsifier of $G$. This definition is a natural analogy of other fault-tolerant subgraphs, such as fault-tolerant spanners and fault-tolerant distance preservers (see, e.g.~\cite{ChechikLPR09,Peleg09,BaswanaCR16,BodwinGPW17,BodwinDPW18}), to the maximum matching problem. Despite being such fundamental objects, quite surprisingly fault-tolerant subgraphs have not previously been studied for the matching problem. 

We complete our discussion of applications of EDCS as a robust sparsifier by showing that it achieves an optimal size fault-tolerant subgraph for the matching problem. 

\begin{result}\label{res:fault-tolerant}
	For any constant $\eps > 0$ and any integer $f \geq 0$, there exists a $(3/2+\eps)$-approximation $f$-tolerant subgraph $H$ of any given graph $G$ with $O(f + n)$ edges in total. 
\end{result}

The number of edges used in our fault-tolerant subgraph in Result~\ref{res:fault-tolerant} is clearly optimal (up to constant factors). In Appendix~\ref{app:optimal-fl}, we show that by modifying the lower bound of~\cite{GoelKK12} in the 
communication model, we can also prove that the approximation ratio of $(3/2)$ is optimal for any $f$-tolerant subgraph with $O(f)$ edges, hence proving that Result~\ref{res:fault-tolerant} is optimal in a strong sense. 
We also show that several natural strategies for this problem cannot achieve better than $2$-approximation, hence motivating our more sophisticated approach toward this problem (see Appendix~\ref{app:fault-tolerant}). 

\smallskip

The qualitative message of our work is clear: \emph{An EDCS is a robust matching sparsifier under all three notions of sparsification described earlier, which leads to simpler and improved algorithms for a wide range of problems
involving sparsification for matching problems in a unified way.}

\subsection*{Overall Proof Strategy}
\newcommand{\gstar}{G^*}
\newcommand{\hstar}{H^*}
Recall that our algorithm in all of the results above is simply to compute an EDCS $H$ of the input graph $G$ (or $G_A$ in the communication problem). The analysis then depends on the specific notion of sparsification at hand, but the same high-
level idea applies to all three cases. In each case, we have an original graph $G$, and then a modified graph $\gstar$ produced by changes to $G$: $\gstar$ is $G_A \cup G_B$ in the communication model, the realized subgraph $G_p$ in the 
stochastic matching, and the graph $G \setminus F$ after adversarially removing edges $F$ in the fault-tolerant matching problem. Let $H$ be the EDCS that our algorithm computes in $G$, and let $\hstar$ be the graph that results from $H$ due to 
the modifications made to $G$. If we could show that $\hstar$ is an EDCS of $\gstar$ then the proof would be complete, since we know that an EDCS is guaranteed to contain an almost $(3/2)$-approximate matching. Unfortunately, in all the three 
problems that we study it might not be the case that $\hstar$ is an EDCS of $\gstar$. Instead in each case we are able to exhibit subgraphs $\tH \subseteq \hstar$ and $\tG \subseteq \gstar$ such that $\tH$ is an EDCS of $\tG$, and size of maximum 
matching of $\tG$ and $\gstar$ differ by at most a $(1+\eps)$ factor. This guarantees an approximation ratio of almost $(3/2)(1 + \eps)$ (precisely what we achieve in all three results above), since the EDCS $\tH$ preserves the maximum matching in 
$\tG$ to within an almost $(3/2)$-approximation and $\tH$ is a subgraph of $H$.

\paragraph{Organization.} The rest of the paper is organized as follows. Section~\ref{sec:prelim} includes notation, simple preliminaries, and existing work on the EDCS. In Section~\ref{sec:edcs}, we present a significantly simpler proof of the fact that an EDCS contains an almost $(3/2)$-approximation matching (originally proved in \cite{BernsteinS16}). Sections~\ref{sec:communication},~\ref{sec:stochastic}, and~\ref{sec:fault-tolerant} 
prove the sparsification properties of the EDCS in, respectively, the one-way communication complexity of matching (Result~\ref{res:communication}), the stochastic matching problem (Result~\ref{res:stochastic}), 
and the fault-tolerant matching problem (Result~\ref{res:fault-tolerant}). These three sections are designed to be self-contained (beside assuming the background in Section~\ref{sec:prelim}) to allow the reader to directly consider the part of most interest. The appendix contains some secondary observations.

%% file: prelim.tex
\section{Preliminaries and Notation}\label{sec:prelim}

\newcommand{\mustar}{\mu^*}

\paragraph{Notation.} For any integer $t \geq 1$, $[t] := \set{1,\ldots,t}$. For a graph $G(V,E)$ and a set of vertices $U \subseteq V$, $N_G(U)$ denotes the neighbors of vertices in $U$ in $G$ and $E_G(U)$ denotes the 
set of edges incident on $U$. Similarly, for a set of edges $F \subseteq E$, $V(F)$ denotes the set of vertices incident on these edges. 
For any vertex $v \in V$, we use $\deg{G}{v}$ to denote the degree of $v \in V$ in $G$ (we may drop the subscript $G$ in these definitions  if it is clear from the context). 
We use $\MM{G}$ to denote the size of the maximum matching in the graph $G$. 

\medskip

Throughout the paper, we use the following two standard variants of the Chernoff bound. 


\begin{proposition}[Chernoff Bound]\label{prop:chernoff}
	Suppose $X_1,\ldots,X_t$ are $t$ independent random variables that take values in $[0,1]$. Let $X := \sum_{i=1}^{t} X_i$ and assume $\Ex\bracket{X} \leq \lambda$. For any $\delta > 0$ and integer $k \geq 1$,
	\begin{align*}
		&\Pr\Paren{\card{X - \Ex\bracket{X}} \geq \delta \cdot \lambda} \leq 2\cdot\exp\Paren{-\frac{\delta^2 \cdot {\lambda}}{3}}\quad \& \quad \Pr\Paren{\card{X - \Ex\bracket{X}} \geq k} \leq 2\cdot\exp\Paren{-\frac{2k^2}{t}}.
	\end{align*} 
\end{proposition}
\noindent
We also need the following basic variant of Lovasz Local Lemma (LLL). 
\begin{proposition}[Lovasz Local Lemma; cf.~\cite{ErdosL73,AlonS04}]\label{prop:lll}
	Let $p \in (0,1)$ and $d \geq 1$. Suppose $\event_1,\ldots,\event_t$ are $t$ events such that $\Pr\paren{\event_i} \leq p$ for all $i \in [t]$ and each $\event_i$ is mutually independent of all but (at most) $d$ other events $\event_j$. 
	If $p \cdot (d+1) < 1/e$ then $\Pr\paren{\cap_{i=1}^{n}\overline{\event_i}} > 0$. 
\end{proposition}

\paragraph{Hall's Theorem.} We use the following standard extension of the Hall's marriage theorem for characterizing maximum matching size in bipartite graphs.

\begin{proposition}[Extended Hall's marriage theorem; cf.~\cite{Hall35}]\label{prop:halls-marriage}
	Let $G(L,R,E)$ be any bipartite graph with $\card{L} = \card{R} = n$. Then,
	$
		\max \Paren{\card{A} - \card{N(A)}} = n - \MM{G}, 
	$
	where $A$ ranges over $L$ or $R$. We refer to such set $A$ as a \emph{witness set}. 
\end{proposition}

Proposition~\ref{prop:halls-marriage} follows from Tutte-Berge formula for matching size in general graphs~\cite{Tutte47,Berge62} or a simple extension of the proof of Hall's marriage theorem itself 

\subsection*{Previously Known Properties of the EDCS}
\label{subsec:prelim-edcs}

\newcommand{\propone}{\textnormal{Property~(P1)}\xspace}
\newcommand{\proptwo}{\textnormal{Property~(P2)}\xspace}

\newcommand{\EDCS}[1]{\ensuremath{\textnormal{EDCS\ensuremath{(#1)}}\xspace}}

\renewcommand{\bA}{\ensuremath{\overline{A}}}

\renewcommand{\bB}{\ensuremath{\overline{B}}}
\newcommand{\bd}{\ensuremath{\bar{d}}}


Recall the definition of an EDCS in Definition \ref{def:edcs}. It is not hard to show that an EDCS always exists as long as $\beta > \beta^-$ (see, e.g.~\cite{AssadiBBMS17}).  For completeness, we repeat the proof in the Appendix \ref{app:edcs-exists}.

\begin{proposition}[cf.~\cite{BernsteinS15,BernsteinS16,AssadiBBMS17}]\label{prop:edcs-exists}
	Any graph $G$ contains an $\EDCS{G,\beta,\beta^-}$ for any parameters $\beta > \beta^-$, which can be found in polynomial time. 
\end{proposition}

The key property of an EDCS, originally proved in ~\cite{BernsteinS15,BernsteinS16}, is that it contains an almost $(3/2)$-approximate matching. 

\begin{lemma}[\!\cite{BernsteinS15,BernsteinS16}]\label{lem:edcs-matching}
	Let $G(V,E)$ be any graph and $\eps < 1/2$ be a parameter. For parameters $\lambda \leq \frac{\eps}{100}$, $\beta \geq 32\lambda^{-3}$, and $\beta^- \geq (1-\lambda) \cdot \beta$, in any subgraph $H:= \EDCS{G,\beta,\beta^-}$, 
	$\MM{G} \leq \paren{\frac{3}{2} + \eps} \cdot \MM{H}$. 
\end{lemma}

Another particularly useful (technical) property of an EDCS is that it ``balances'' the degree of vertices and their neighbors
in the EDCS; this property is implicit in~\cite{BernsteinS15} but we explicitly state and prove it here as it shows a main distinction in the properties of EDCS compared to more standard (and less robust) subgraphs in this context such as $b$-matchings. 

\begin{proposition}\label{prop:edcs-balance}
	Let $H := \EDCS{G,\beta,\beta^-}$ and $U$ be any subset of vertices. If average degree of $U$ in $H$ is $\bd$ then the average degree of $N_H(U)$ from edges incident on $U$ is at most $\beta-\bd$. 
\end{proposition}
\begin{proof}
	Let $H'$ be a subgraph of $H$  containing the edges incident on $U$. 	Let $W := N_{H'}(U) = N_H(U)$ and $E' = E_H(U,W) = E_{H'}(U,W)$. 
	We are interested in upper bounding the quantity $\card{E'}/\card{W}$. Firstly, by $\propone$ of EDCS, we have that
	$\sum_{(u,v) \in E'} \deg{H'}{u} + \deg{H'}{v} \leq \beta \cdot \card{E'}.$ We write the LHS in this equation as:
	\begin{align*}
	\sum_{(u,v) \in E'} \deg{H'}{u} + \deg{H'}{v} &= \sum_{u \in U} (\deg{H'}{u})^2 + \sum_{w \in W} (\deg{H'}{w})^2 
	\geq \sum_{u \in U} (\frac{\card{E'}}{\card{U}})^2 + \sum_{w \in W} (\frac{\card{E'}}{\card{W}})^2 \tag{as $\sum_{u} \deg{H'}{u} = \sum_{w}\deg{H'}{w} = \card{E'}$ and each is minimized when the summands are equal.} \\
	&= \card{E'} \cdot \paren{\bd + \card{E'}/\card{W}}.
	\end{align*}
	\noindent
	By plugging in this bound in LHS above, we obtain $\card{E'}/\card{W} \leq \beta - \bd$, finalizing the proof. 
\end{proof}

%% file: edcs.tex
\section{A Simpler Proof of the Key Property of an EDCS}\label{sec:edcs}

In this section we provide a much simpler proof of the key property that an EDCS contains an almost $(3/2)$-approximate matching.
This lemma was previously used in~\cite{BernsteinS15,BernsteinS16,AssadiBBMS17}.
Our proof is self-contained to this section, and for general graphs, our new proof even improves the dependence of $\beta$ on parameter $\lambda$
from $1/\lambda^3$ to (roughly) $1/\lambda^2$, thus allowing for an even sparser EDCS.

The proof contains two steps. We first give a simple and streamlined proof that an EDCS contains a $(3/2)$-approximate matching in bipartite graphs. Our proof 
in this part is similar to~\cite{BernsteinS15} but instead of modeling matchings as flows and using cut-flow duality, we directly work with matchings by using Hall's theorem. The main part of the proof however is to extend this result to general graphs. 
For this, we give a simple reduction that extends the result on bipartite graphs to general graphs by taking advantage of the ``robust'' nature of EDCS. This allows us to bypass the complicated 
arguments in~\cite{BernsteinS16} specific to non-bipartite graphs and to obtain the result directly from the one for bipartite graphs (the paper of \cite{BernsteinS16} explicitly acknowledges the complexity of the proof and asks for a more ``natural" approach).

\subsection*{A Slightly Simpler Proof for Bipartite Graphs}
Our new proof should be compared to Lemma 2 in Section 4.1 of the Arxiv version of \cite{BernsteinS15}.

\begin{lemma}\label{lem:bipartite-edcs-matching}
	Let $G(L,R,E)$ be any bipartite graph and $\eps < 1/2$ be a parameter. For $\lambda \leq \frac{\eps}{4}$, $\beta \geq 2\lambda^{-1}$, and $\beta^- \geq (1-\lambda) \cdot \beta$, in any subgraph $H:= \EDCS{G,\beta,\beta^-}$, 
	$\MM{G} \leq \paren{\frac{3}{2} + \eps} \cdot \MM{H}$. 
\end{lemma}
\begin{proof}
	Fix any $H:=\EDCS{G,\beta,\beta^-}$ and let $A$ be any of its witness sets in extended Hall's marriage theorem of Proposition~\ref{prop:halls-marriage} and $B:= N_H(A)$. Without loss of generality, let us assume $A$ is a subset of $L$. 
	Define $\bA := L \setminus A$, $\bB := R \setminus B$ (see Figure~\ref{fig:bipartite-edcs-matching}). By Proposition~\ref{prop:halls-marriage}, 
	\begin{align}
	\card{\bA} + \card{B} = n-(\card{A}-\card{B}) \leq n-(n-\MM{H}) = \MM{H}. \label{eq:H-matching}
	\end{align}
	On the other hand, since $G$ has a matching of size $\MM{G}$, we need to have a matching $M$ of size $(\MM{G}-\MM{H})$ between $A$ and $\bB$ as otherwise by Proposition~\ref{prop:halls-marriage}, $A$ would be a witness set in $G$ that implies the
	maximum matching of $G$ is smaller than $\MM{G}$ (to see why the set of edges between $A$ and $\bB$ is a matching simply apply Proposition~\ref{prop:halls-marriage} to a
	subgraph of $G$ containing only a maximum matching of $G$). Let $S \subseteq A \cup \bB$ be the end points of this matching (see Figure~\ref{fig:bipartite-edcs-matching}). 
	As edges in $M$ are all missing from $H$, by $\proptwo$ of EDCS $H$, we have that, 
	\begin{align}
		\sum_{v \in S} \deg{H}{v}  = \sum_{(u,v) \in M} (\deg{H}{u} + \deg{H}{v}) \geq (\MM{G}-\MM{H}) \cdot \beta^-. \label{eq:H-edge}
	\end{align}
	Consequently, as $\card{S} = 2(\MM{G} - \MM{H})$, the average degree of $S$ is $\geq \beta^-/2$. As such, by Proposition~\ref{prop:edcs-balance}, the average degree of 
	of $N_H(S)$ (from $S$) is at most $\beta - \beta^-/2 \leq (1+\lambda) \beta/2$. Finally, note that $N_H(S) \subseteq \bA \cup B$ as there are no edges between $A$ and $\bB$ in $H$, and hence by Eq~(\ref{eq:H-matching}), $\card{N_H(S)} \leq \MM{H}$. 
	By double counting the number of edges between $S$ and $N_H(S)$, i.e., $E_H(S)$: 
	\begin{align*}
	\card{E_H(S)} &\geq \card{S} \cdot \beta^-/2 \geq 2(\MM{G} - \MM{H}) \cdot \beta^-/2, \\
	\card{E_H(S)} &\leq \card{N_H(S)} (1+\lambda)\beta/2 \leq \MM{H} \cdot (1+\lambda)\beta/2.
	\end{align*}
	This implies that, 
	\begin{align*}
		2\MM{G} \leq 2\MM{H} + \MM{H} \cdot (1+\lambda)(\beta/2) \cdot (2/\beta^-) \leq 3\MM{H} \cdot \frac{1+\lambda}{1-\lambda} \leq 3\MM{H} (1+\eps).
	\end{align*}
	Reorganizing the terms above, finalizes the proof. 
\end{proof}

\input{edcs-matching-bipartite}

\subsection*{A Much Simpler Proof for Non-bipartite Graphs}

Our new proof in this part should be compared to Lemma 5.1 on page 699 in~\cite{BernsteinS16}: see Appendix B of their paper for the full proof, as well Section 4 for an additional auxiliary claim needed. 

\begin{lemma}\label{lem:general-edcs-matching}
	Let $G(V,E)$ be any graph and $\eps < 1/2$ be a parameter. For $\lambda \leq \frac{\eps}{32}$, $\beta \geq 8\lambda^{-2}\log{(1/\lambda)}$, and $\beta^- \geq (1-\lambda) \cdot \beta$, in any subgraph $H:= \EDCS{G,\beta,\beta^-}$, 
	$\MM{G} \leq \paren{\frac{3}{2} + \eps} \cdot \MM{H}$. 
\end{lemma}
\begin{proof}
	The proof is based on the probabilistic method and Lovasz Local Lemma. Let $\Mstar$ be a maximum matching of size $\MM{G}$ in $G$. Consider the following randomly chosen bipartite subgraph $\tG(L,R,\tE)$ of $G$ with respect to $\Mstar$, where
	$L \cup R = V$: 
	\begin{itemize}
		\item For any edge $(u,v) \in \Mstar$, with probability $1/2$, $u$ belongs to $L$ and $v$ belongs to $R$, and with probability $1/2$, the opposite (the choices between different edges of $\Mstar$ are independent). 
		\item For any vertex $v \in V$ not matched by $\Mstar$, we assign $v$ to $L$ or $R$ uniformly at random (again, the choices are independent across vertices). 
		\item The set of edges in $\tE$ are all edges in $E$ with one end point in $L$ and the other one in $R$. 
	\end{itemize}
	
	Define $\tH := H \cap \tG$. We argue that as $H$ is an EDCS for $G$, $\tH$ also remains an EDCS for $\tG$ with non-zero probability. Formally, 
	\begin{claim}\label{clm:tH-is-EDCS}
		$\tH$ is an $\EDCS{\tG,\tbeta,\tbeta^-}$ for $\tbeta = (1+4\lambda) \beta/2$ and $\tbeta^- = (1-5\lambda)\beta^-/2$ with probability strictly larger than zero (over the randomness of $\tG$).   
	\end{claim}
	
	Before we prove Claim~\ref{clm:tH-is-EDCS}, we argue why it implies Lemma~\ref{lem:general-edcs-matching}. Let $\tG$ be chosen such that $\tH$ is an $\EDCS{\tG,\tbeta,\tbeta^-}$ for parameters in Claim~\ref{clm:tH-is-EDCS} (by Claim~\ref{clm:tH-is-EDCS}, such a choice of $\tG$ always exist). By construction of $\tG$, $\Mstar \subseteq \tE$ and hence $\MM{\tG} = \MM{G}$. On the other hand, $\tG$ is now a bipartite 
	graph and $\tH$ is its EDCS with appropriate parameters. We can hence apply Lemma~\ref{lem:bipartite-edcs-matching} and obtain that $\MM{\tG} \leq (3/2 + \eps) \MM{\tH}$. As $\tH \subseteq H$, $\MM{\tH} \leq \MM{H}$,
	and hence $(\MM{\tG}=)\MM{G} \leq (3/2 + \eps)\MM{H}$, proving the assertion in the lemma statement. It thus only remains to prove Claim~\ref{clm:tH-is-EDCS}. 
	
	\begin{proof}[Proof of Claim~\ref{clm:tH-is-EDCS}]
		Fix any vertex $v \in V$, let $d_v := \deg{H}{v}$ and $N_H(v):= \set{u_1,\ldots,u_{d_v}}$ be the neighbors of $v$ in $H$. Let us assume $v$ is chosen in $L$ in $\tG$ (the other case is symmetric). Hence, degree of $v$ in $\tH$ is 
		exactly equal to the number of vertices in $N_H(v)$ that are chosen in $R$. As such, by construction of $\tG$, $\Ex\bracket{\deg{\tH}{v}} = d_v/2$ ($+1$ iff $v$ is incident on $\Mstar \cap H$). Moreover, if two vertices $u_i,u_j$ in $N_H(v)$ are matched by $\Mstar$, then
		exactly one of them appears as a neighbor to $v$ in $\tH$ and otherwise the choices are independent. Hence, by Chernoff bound (Proposition~\ref{prop:chernoff}), 
		\begin{align*}
			\Pr\paren{\card{\deg{\tH}{v} - d_v/2} \geq \lambda \cdot \beta} \leq \exp\paren{-\frac{2\lambda^2 \cdot \beta^2}{\beta}} \leq \exp\paren{-4\log{\beta}} \leq \frac{1}{\beta^4} \tag{as $\beta \geq 8\lambda^{-2}\log{(1/\lambda)}$ and 
			hence $\beta \geq 2\lambda^{-2}\cdot\log{\beta}$}. 
		\end{align*}
		
		Define $\event_v$ as the event that $\card{\deg{\tH}{v} - d_v/2} \geq \lambda \cdot \beta$. Note that $\event_v$ depends only on the choice of vertices in $N_H(v)$ and hence can depend on at most $\beta^2$ other events $\event_u$ 
		for vertices $u$ which are neighbors to $N_H(v)$ (recall that for all $u \in V$, $\deg{H}{u} \leq \beta$ in $H$ by $\propone$ of EDCS). As such, we can apply Lovasz Local Lemma (Proposition~\ref{prop:lll}) to argue that 
		with probability strictly more than zero, $\cap_{v \in V} \overline{\event_v}$ happens. In the following, we condition on this event and argue that in this case, $\tH$ is an EDCS of $\tG$ with appropriate parameters. To do this, we only 
		need to prove that both $\propone$ and $\proptwo$ hold for the EDCS $\tH$ (with the choice of $\tbeta$ and $\tbeta^-$).  
		
		We first prove $\propone$ of EDCS $\tH$. Let $(u,v)$ be any edge in $\tH$. 
		By events $\overline{\event_v}$ and $\overline{\event_u}$, 
		\begin{align*}
			\deg{\tH}{u} + \deg{\tH}{v} \leq \frac{1}{2} \cdot \paren{\deg{H}{u} + \deg{H}{v}} + 2\lambda \beta \leq \beta/2 + 2\lambda\beta = (1+4\lambda)\cdot \beta/2,
		\end{align*}
		where the second inequality is by \propone of EDCS $H$ as $(u,v)$ belongs to $H$ as well. We now prove $\proptwo$ of EDCS $\tH$. Let $(u,v)$ be any edge in $\tG \setminus \tH$. Again, by 
		$\overline{\event_v}$ and $\overline{\event_u}$, 
		\begin{align*}
			\deg{\tH}{u} + \deg{\tH}{v} \geq \frac{1}{2} \cdot \paren{\deg{H}{u} + \deg{H}{v}} - 2\lambda \beta \geq \beta^-/2 - 2\lambda(1-\lambda)\beta^- \geq (1-5\lambda)\cdot \beta/2,
		\end{align*}
		where the second inequality is by \proptwo of EDCS $H$ as $(u,v) \in G \setminus H$.  \Qed{Claim~\ref{clm:tH-is-EDCS}}
		
	\end{proof}
	\noindent
	Lemma~\ref{lem:general-edcs-matching} now follows immediately from Claim~\ref{clm:tH-is-EDCS} as argued above. 	\Qed{Lemma~\ref{lem:general-edcs-matching}}
	
\end{proof}

%% file: edcs-matching-bipartite.tex
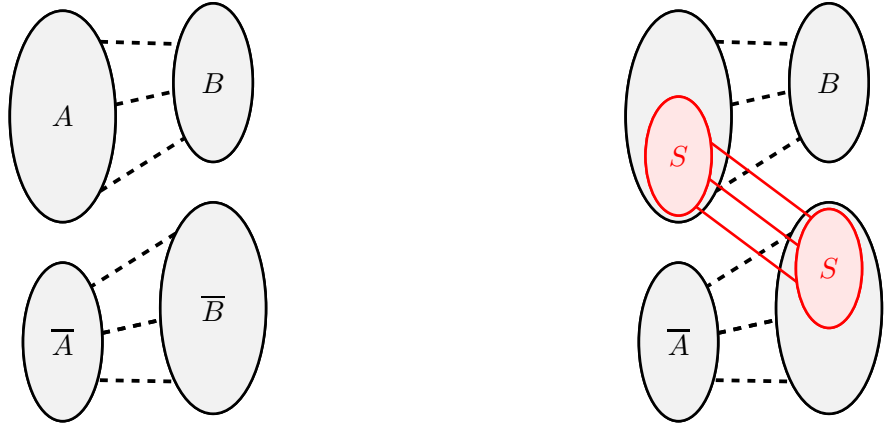
\begin{figure}[t]
    \centering
    \subcaptionbox{$A$ and $B:= N_H(A)$ form a Hall's theorem witness set in the EDCS $H$ and $\card{\bA \cup B} \leq \MM{H}$. 
   }[0.45\textwidth]{
    \begin{tikzpicture}[ auto ,node distance =1cm and 2cm , on grid , semithick , state/.style ={ circle ,top color =white , bottom color = white , draw, black , text=black}, every node/.style={inner sep=0,outer sep=0}]

\node[ellipse, draw, black, fill=black!5, text=black, minimum height=80pt, minimum width=40pt, line width=1pt](A){$A$};

\node[ellipse, draw, black, fill=black!5, text=black, minimum height=60pt, minimum width=30pt, line width=1pt]  (bA) [below = 3cm of A]{$\bA$};

\node[ellipse, draw, black, fill=black!5, text=black, minimum height=60pt, minimum width=30pt, line width=1pt] (B) [above right = 0.45cm and 2cm of A]{$B$};

\node[ellipse, draw, black, fill=black!5, text=black, minimum height=80pt, minimum width=40pt, line width=1pt] (bB) [above right= 0.45cm and 2cm of bA]{$\bB$};


\draw[line width=0.5mm, dashed, black] (A) to (B);

\node (abB) [above=0.5cm of B]{};
\draw[line width=0.5mm, dashed, black] (A.north east) to (abB);

\node (beB) [below=0.5cm of B]{};
\draw[line width=0.5mm, dashed, black] (A.south east) to (beB);


\draw[line width=0.5mm, dashed, black] (bA) to (bB);

\node (abA) [above=0.5cm of bA]{};
\draw[line width=0.5mm, dashed, black] (bB.north west) to (abA);

\node (beA) [below=0.5cm of bA]{};
\draw[line width=0.5mm, dashed, black] (bB.south east) to (beA);

\draw[line width=0.5mm, dashed, black] (bA) to (bB);


\node[ellipse, draw, black, fill=black!5, text=black, minimum height=80pt, minimum width=40pt, line width=1pt](A){$A$};

\node[ellipse, draw, black, fill=black!5, text=black, minimum height=60pt, minimum width=30pt, line width=1pt]  (bA) [below = 3cm of A]{$\bA$};

\node[ellipse, draw, black, fill=black!5, text=black, minimum height=60pt, minimum width=30pt, line width=1pt] (B) [above right = 0.45cm and 2cm of A]{$B$};

\node[ellipse, draw, black, fill=black!5, text=black, minimum height=80pt, minimum width=40pt, line width=1pt] (bB) [above right= 0.45cm and 2cm of bA]{$\bB$};

\end{tikzpicture}
    
    }
    \hspace{0.5cm}
    \subcaptionbox{There is a matching of size $\MM{G}-\MM{H}$ between $A$ and $\bB$ (i.e., the set $S$) in $G \setminus H$. 
     }[0.45\textwidth]{
        \begin{tikzpicture}[ auto ,node distance =1cm and 2cm , on grid , semithick , state/.style ={ circle ,top color =white , bottom color = white , draw, black , text=black}, every node/.style={inner sep=0,outer sep=0}]
\node[ellipse, draw, black, fill=black!5, text=black, minimum height=80pt, minimum width=40pt, line width=1pt](A){$A$};

\node[ellipse, draw, black, fill=black!5, text=black, minimum height=60pt, minimum width=30pt, line width=1pt]  (bA) [below = 3cm of A]{$\bA$};

\node[ellipse, draw, black, fill=black!5, text=black, minimum height=60pt, minimum width=30pt, line width=1pt] (B) [above right = 0.45cm and 2cm of A]{$B$};

\node[ellipse, draw, black, fill=black!5, text=black, minimum height=80pt, minimum width=40pt, line width=1pt] (bB) [above right= 0.45cm and 2cm of bA]{$\bB$};


\draw[line width=0.5mm, dashed, black] (A) to (B);

\node (abB) [above=0.5cm of B]{};
\draw[line width=0.5mm, dashed, black] (A.north east) to (abB);

\node (beB) [below=0.5cm of B]{};
\draw[line width=0.5mm, dashed, black] (A.south east) to (beB);


\draw[line width=0.5mm, dashed, black] (bA) to (bB);

\node (abA) [above=0.5cm of bA]{};
\draw[line width=0.5mm, dashed, black] (bB.north west) to (abA);

\node (beA) [below=0.5cm of bA]{};
\draw[line width=0.5mm, dashed, black] (bB.south east) to (beA);

\draw[line width=0.5mm, dashed, black] (bA) to (bB);


\node[ellipse, draw, black, fill=black!5, text=black, minimum height=80pt, minimum width=40pt, line width=1pt](A){$A$};

\node[ellipse, draw, black, fill=black!5, text=black, minimum height=60pt, minimum width=30pt, line width=1pt]  (bA) [below = 3cm of A]{$\bA$};

\node[ellipse, draw, black, fill=black!5, text=black, minimum height=60pt, minimum width=30pt, line width=1pt] (B) [above right = 0.45cm and 2cm of A]{$B$};

\node[ellipse, draw, black, fill=black!5, text=black, minimum height=80pt, minimum width=40pt, line width=1pt] (bB) [above right= 0.45cm and 2cm of bA]{$\bB$};


\node[ellipse,  draw,  red, fill=red!10,  minimum height=45pt, minimum width = 25pt, rounded corners=2mm, line width=1pt]  (SL) [below = 15pt of A]{$S$};

\node[ellipse, draw, red, fill=red!10, minimum height=45pt, minimum width = 25pt, rounded corners=2mm, line width=1pt]  (SR) [above = 15pt of bB]{$S$};

\node[state,white](x1) [above=0.5cm of SL] {};
\node[state,white] (y1) [above=0.5cm of SR] {};

\draw[red, line width=1pt](x1) to (y1);

\node[state,white] (x2) [below=0.5cm of x1] {};
\node[state,white] (y2) [below=0.5cm of y1] {};

\draw[red, line width=1pt] (x2) to (y2);

\node[state,white] (x3) [below=0.5cm of x2] {};
\node[state,white] (y3) [below=0.5cm of y2] {};

\draw[red, line width=1pt] (x3) to (y3);

\node[ellipse,  draw,  red, fill=red!10,  minimum height=45pt, minimum width = 25pt, rounded corners=2mm, line width=1pt]  (SL) [below = 15pt of A]{$S$};

\node[ellipse, draw, red, fill=red!10, minimum height=45pt, minimum width = 25pt, rounded corners=2mm, line width=1pt]  (SR) [above = 15pt of bB]{$S$};

\end{tikzpicture}

    }
    \caption{The partitioning of vertices used in the proof of Lemma~\ref{lem:bipartite-edcs-matching}.}
    \label{fig:bipartite-edcs-matching}
\end{figure}

%% file: communication.tex
\section{One-Way Communication Complexity of Matching}\label{sec:communication}

In the one-way communication model, Alice and Bob are given graphs $G_A(V,E_A)$ and $G_B(V,E_B)$, respectively, and the goal is for Alice to send a small message to Bob
such that Bob can output a large approximate matching in $E_A \cup E_B$. In this section, we show that if Alice communicates an appropriate EDCS of $G_A$, then Bob is able to output an almost $(3/2)$-approximate matching. 

\begin{theorem}[Formalizing Result~\ref{res:communication}]\label{thm:communication}
	There exists a deterministic poly-time one-way communication protocol that given any $\eps > 0$, computes a $(3/2+\eps)$-approximation to maximum matching using $O(\frac{n \cdot \log{(1/\eps)}}{\eps^2})$ communication from Alice to Bob. 
\end{theorem}

Theorem~\ref{thm:communication} is based on the following protocol: 

\begin{tbox}

\underline{A one-way communication protocol for maximum matching.} 

\begin{enumerate}
\item Alice computes $H:=\EDCS{G_A,\beta,\beta-1}$ for $\beta := 32 \cdot \eps^{-2} \cdot \log{(1/\eps)}$ and sends it to Bob. 
\item Bob computes a maximum matching in $H \cup G_B$ and outputs it as the solution.
\end{enumerate}
\end{tbox}

By Proposition~\ref{prop:edcs-exists}, the EDCS $H$ computed by Alice always exists and can be found in polynomial time.
Moreover, by $\propone$ of EDCS $H$, the total number of edges (and hence the message size) sent by Alice is $O(n\beta)$. 
We now prove the correctness of the protocol which concludes the proof of Theorem~\ref{thm:communication}. 

\begin{lemma}\label{lem:communication}
	$\MM{G_A \cup G_B} \leq (3/2+\eps) \cdot \mu(H \cup G_B)$. 
\end{lemma}
\begin{proof}
	Let $\Mstar$ be a maximum matching in $G_A \cup G_B$ and $\Mstar_A$ and $\Mstar_B$ be its edges in $G_A$ and $G_B$, respectively. 
	Let $\tG := G_A \cup \Mstar_B$ and note that $\MM{\tG} = \MM{G}$ simply because $\Mstar$ belongs to $\tG$. Define the following subgraph $\tH \subseteq H \cup \Mstar_B$ (and hence $\subseteq H \cup G_B$): 
	$\tH$ contains all edges in $H$ and any edge $(u,v) \in \Mstar_B$ such that $\deg{H}{u} + \deg{H}{v} \leq \beta$. In the following, we prove that $(\MM{G}=)\MM{\tG} \leq (3/2+\eps) \cdot \MM{\tH}$, which finalizes the proof
	as $\MM{\tH} \leq \MM{H \cup G_B}$.
	
	We show that $\tH$ is an $\EDCS{\tG,\beta+2,\beta-1}$ and apply Lemma~\ref{lem:general-edcs-matching} to argue that $\tH$ contains a $(3/2)$-approximate matching of $\tG$. 
	We prove the EDCS properties of $\tH$ using the fact that for $v \in V$, $\deg{\tH}{v} \in \set{\deg{H}{v},\deg{H}{v}+1}$ as $\tH$ is obtained by adding a matching ($\subseteq \Mstar_B$) to $H$.
	\begin{itemize}
	\vspace{-0.2cm}
	\item $\propone$ of EDCS $\tH$: For an edge $(u,v) \in \tH$,
	\begin{align*}
		&\textnormal{if $(u,v) \in H$ then:} \quad &&\deg{\tH}{u} + \deg{\tH}{v} \leq \deg{H}{u} + \deg{H}{v} + 2 \leq \beta + 2, \tag{by $\propone$ of EDCS $H$ of $G_A$} \\
		&\textnormal{if $(u,v) \in \Mstar_B$ then:} \quad &&\deg{\tH}{u} + \deg{\tH}{v} \leq \deg{H}{u} + \deg{H}{v} + 2 \leq \beta + 2. \tag{as $(u,v) \in \Mstar_B$ is inserted to $\tH$ iff $\deg{H}{u} + \deg{H}{v} \leq \beta$}
	\end{align*}
		\vspace{-0.25cm}
	\item $\proptwo$ of EDCS $\tH$: For an edge $(u,v) \in \tG \setminus \tH$, 
	\begin{align*}
		&\textnormal{if $(u,v) \in G_A \setminus H$ then:} \qquad &&\deg{\tH}{u} + \deg{\tH}{v} \geq \deg{H}{u} + \deg{H}{v}  \geq \beta-1, \tag{by $\proptwo$ of EDCS $H$ of $G_A$} \\
		&\textnormal{if $(u,v) \in \Mstar_B \setminus \tH$ then:} \qquad &&\deg{\tH}{u} + \deg{\tH}{v} \geq \deg{H}{u} + \deg{H}{v} > \beta. \tag{as $(u,v) \in \Mstar_B$ is not inserted to $\tH$ iff $\deg{H}{u} + \deg{H}{v} > \beta$}
	\end{align*}
	\end{itemize}
	
	As such, $\tH$ is an $\EDCS{\tG,\beta+2,\beta-1}$. By Lemma~\ref{lem:general-edcs-matching} and the choice of parameter $\beta$, we obtain that $\MM{\tG} \leq (3/2+\eps)\cdot\MM{\tH}$, finalizing the proof. 
\end{proof}

%% file: stochastic.tex
\section{The Stochastic Matching Problem}\label{sec:stochastic}

\newcommand{\hstarp}{\tH_p}
\newcommand{\gstarp}{\tG_p}

Recall that in the stochastic matching problem, the goal is to compute a bounded-degree subgraph $H$ of a given graph $G$, such 
that $\Ex\bracket{\MM{H_p}}$ is a good approximation of $\Ex\bracket{\MM{G_p}}$, where $G_p$ is a realization of $G$ (i.e a subgraph where every edge is sampled with probability $p$), and $H_p = H \cap G_p$. 
In this section, we formalize Result~\ref{res:stochastic} by proving the following theorem. 

\begin{theorem}[Formalizing Result~\ref{res:stochastic}]\label{thm:stochastic}
	There exists a deterministic poly-time algorithm that given a graph $G(V,E)$ and parameters $\eps,p > 0$ with $\eps < 1/4$, computes a subgraph $H(V,E_H)$ of $G$ with maximum degree $O(\frac{\log{(1/\eps p)}}{\eps^2 \cdot p})$ such that the ratio of the expected size of a maximum matching in 
	realizations of $G$ to realizations of $H$ is at most $(3/2+\eps)$, i.e., $\Ex\bracket{\MM{G_p}} \leq (3/2+\eps) \cdot \Ex\bracket{\MM{H_p}}$.   
\end{theorem}

We note that while in Theorem~\ref{thm:stochastic}, we state the bound in expectation, the same result also holds with high probability as long as $\mu(G) = \omega(1/p)$ (i.e., just barely more than a constant), by
concentration of maximum matching size in edge-sampled subgraphs (see, e.g.~\cite{AssadiBBMS17}, Lemma 3.1). 
The algorithm in Theorem~\ref{thm:stochastic} simply computes an EDCS of the input graph as follows: 

\begin{tbox}

\underline{An algorithm for the stochastic matching problem.}

\medskip

Output the subgraph $H:=\EDCS{G,\beta,\beta-1}$ for $\beta := \frac{C \log{(1/\eps p)}}{\eps^2 p}$, for large enough constant $C$.
\end{tbox}

By Proposition~\ref{prop:edcs-exists}, the EDCS $H$ in the above algorithm always exists and can be found in polynomial time.
Moreover, by $\propone$ of EDCS $H$, the total number of edges in this subgraph is $O(n\beta)$. We now prove the bound on the approximation ratio which concludes the proof of Theorem~\ref{thm:stochastic} (by re-parametrizing 
$\eps$ to be a constant factor smaller). 

\begin{lemma}\label{lem:stochastic}
	Let $H_p := H \cap G_p$ denote a realization of $H$; then $\Ex\bracket{\MM{G_p}} \leq (3/2+O(\eps)) \cdot \Ex\bracket{\MM{H_p}}$ where the randomness is taken over the realization $G_p$ of $G$. 
\end{lemma}

Suppose first that $H_p$ were an EDCS of $G_p$; we would be immediately done in this case  as we could have applied Lemma~\ref{lem:general-edcs-matching} directly and prove Lemma~\ref{lem:stochastic}. Unfortunately, however, this 
might not be the case. Instead, we exhibit subgraphs $\hstarp \subseteq H_p$ and $\gstarp \subseteq G_p$ with the following properties:
\begin{enumerate}
	\item\label{prop:tGp} $\expect{\mu(G_p)} \leq (1+\eps)\expect{\mu(\gstarp)}$, where the expectation is taken over the realization $G_p$ of $G$.
	\item\label{prop:tHp} $\hstarp$ is an $\EDCS{G,(1+\eps)p \cdot \beta, (1-2\eps)p \cdot \beta}$ for $\gstarp$.
\end{enumerate}
\noindent
Showing these properties concludes the proof of Lemma~\ref{lem:stochastic}, as for the EDCS in item~(\ref{prop:tHp}) above, we have
$\frac{(1+\eps)p \cdot \beta}{(1-2\eps) \cdot p\beta} = 1+O(\eps)$, so by Lemma \ref{lem:general-edcs-matching}
we get that $\mu(\gstarp) \leq (3/2 + O(\eps)) \cdot \mu(\hstarp)$. Combining this with item~(\ref{prop:tGp}) then concludes  
$\expect{\mu(G_p)} \leq (1+\eps) \cdot (3/2+\eps)\expect{\mu(H_p)}$. 

It now remains to exhibit $\hstarp$ and $\gstarp$ that satisfy the main properties stated above. Note that for any vertex $v \in V$, we have 
$\expect{\deg{H_p}{v}} = p \cdot {\deg{H}{v}}$ by definition of a realization $G_p$ (and hence $H_p$).
We now want to separate out vertices that deviate significantly from this expectation.

\begin{definition} Let $V^+ \subseteq V$ contain all vertices $v$ for which
$\deg{H_p}{v} >  p \cdot \deg{H}{v} + \eps p \beta / 2$. 
Similarly, let $V^-$ contain all vertices $v$ such that
$\deg{H_p}{v} <  p \cdot \deg{H}{v} -  \eps p \beta / 2$
OR there exists an edge $(v,w) \in H$ such that $w \in V^+$, i.e., if $v$ is neighbor to $V^+$.
\end{definition}

\begin{claim}
	\label{claim:stochastic}
	 $\expect{|V^+|} \leq \eps^{7} p^{7}\mu(G)$ and $\expect{|V^-|} \leq \eps^{4}p^{4}\mu(G)$, where the expectation is over the realization $G_p$ of $G$. As we a result we also have $\expect{|V^+| + |V^-|} \leq \eps^{3}p^{3}\mu(G)$.
\end{claim}
\noindent
Before proving this claim, let us consider why it completes the proof of Lemma~\ref{lem:stochastic}.

\input{stochastic-matching}

\begin{proof}[Proof of Lemma \ref{lem:stochastic} (assuming Claim \ref{claim:stochastic})]
To prove Lemma \ref{lem:stochastic} it is enough to show the existence of subgraphs $\gstarp$ and $\hstarp$ that satisfy the properties above.
We define $\gstarp$ as follows: the vertex set is $V$ and the edge-set is  the same as $G_p$, except we remove all edges incident to $V^+$ and all edges $(u,v) \notin H$ that are incident to $V^-$. We define $\hstarp$ to be the subgraph of $H_p$ induced by the vertex set $V \setminus V^+$, that is, $\hstarp$ contains all edges of $H_p$ except those incident to $V^+$; see Figure~\ref{fig:stochastic-matching}.

For item~(\ref{prop:tGp}), note that $\gstarp$ differs
from $G_p$ by vertices in $V^+ \cup V^-$, so $\mu(\gstarp) \geq \mu(G_p)- |V^+| - |V^-|$. 
It is also clear that 
$\expect{\mu(G_p)} \geq p \cdot  \mu(G)$ (as each edge in $G$ is sampled w.p. $p$ in $G_p$). 
By Claim \ref{claim:stochastic}, 
$$\expect{\mu(\gstarp)} \geq \expect{\mu(G_p)} - \expect{|V^+| - |V^-|} \geq \expect{\mu(G_p)} - p^{3}\eps^{3}\mu(G) \geq (1 - \eps^{3})\expect{\mu(G_p)}.$$
The above equation then implies the desired $\expect{\mu(G_p)} \leq (1+\eps) \expect{\mu(\gstarp)}$.

For item~(\ref{prop:tHp}), let us verify \propone and \proptwo for EDCS $\hstarp$ of $\gstarp$. Neither $\hstarp$ nor $\gstarp$ have any edge 
incident on $V^+$ and hence we can ignore these vertices entirely. Thus, for all vertices $v$ we have $\deg{\hstarp}{v} \leq p \cdot \deg{H}{v} + \eps p \beta / 2$, and
for all $v \notin V^-$ we have  $\deg{\hstarp}{v} \geq p \cdot \deg{H}{v} - \eps p \beta / 2$. Moreover, recall that $\gstarp \setminus \hstarp$ contains
no edges incident to $V^-$. As such,

\begin{itemize}
	\item $\propone$ of EDCS $\hstarp$: For an edge $(u,v) \in \hstarp$,
	\begin{equation*}
	\deg{\hstarp}{u} + \deg{\hstarp}{v} \leq p \cdot \deg{H}{u} + p \cdot \deg{H}{v} + \eps p \beta \leq (1 + \eps) p \beta. \tag{by $\propone$ of EDCS $H$ of $G$} 
	\end{equation*}
	\item $\proptwo$ of EDCS $\hstarp$: For any edge $(u,v) \in \gstarp \setminus \hstarp$, we have $u,v \notin V^-$ so: 
	\begin{equation*}	
	\deg{\hstarp}{u} + \deg{\hstarp}{v} \geq p \cdot \deg{H}{u} + p \cdot \deg{H}{v} - \eps p \beta \geq (1 - 2\eps) p \beta. \tag{by $\proptwo$ of EDCS $H$ of $G$} 
	\end{equation*}
\end{itemize}
This concludes the proof of Lemma~\ref{lem:stochastic} (assuming Claim \ref{claim:stochastic}). \Qed{Lemma~\ref{lem:stochastic}} 

\end{proof}

All that remains is to prove Claim \ref{claim:stochastic}.

\begin{proof}[Proof of Claim \ref{claim:stochastic}]
Let us start by bounding the size of $V^+$. Consider any vertex $v \in V$. We know that $\deg{H}{v} \leq \beta$. Each edge then has probability $p$ of appearing
in $H_p$, so  $\expect{\deg{H_p}{v}} = p \cdot \deg{H}{v} \leq p \beta$. By the multiplicative Chernoff bound in Proposition \ref{prop:chernoff} with $\lambda = p \beta$: 
$$\Pr[v \in V^+] = \Pr[\deg{H_p}{v} \geq p \cdot \deg{H}{v} + \eps p \beta / 2] \leq e^{-O(\eps^2 p \beta)} \leq e^{-O(\log(\eps^{-1}p^{-1}))} \leq K^{-2} \eps^{10}p^{10},$$
where $K$ is a large constant and the last two inequalities follow from the fact that we set $\beta := \frac{C \log{(1/\eps p)}}{\eps^2 p}$, for large enough constant $C$. (Note that since constant $C$ is in the exponent, we can easily set $C$ large enough to achieve the final probability with a constant $K > C$.)
This probability bound shows that $\expect{|V^+|} \leq nK^{-2}\eps^{10}p^{10}$, but that is not quite good enough since we want a dependence on $\mu(G)$ instead of on $n$.
To achieve this, we observe that the total number of edges in $H$ is at most $\beta\mu(G)$: the reason is that $G$ has a vertex cover of size at most $2\mu(G)$,
and all vertices in $H$ have degree at most $\beta$ (by $\propone$ of EDCS $H$). There are thus at most $2\beta\mu(G)$ vertices that have non-zero degree in $H$, each of which
has at most a $\eps^{10}p^{10}$ probability of being in $V^+$; all vertices with zero degree in $H$ are clearly not in $V^+$ by definition. We thus have
$\expect{|V^+|} \leq 2\beta\mu(G) \cdot K^{-2}\eps^{10}p^{10} \leq K^{-1}\eps^{7}p^{7}\mu(G)$, where in the last inequality we use that $K > C$.

Let us now consider $V^-$. First let us bound the number of vertices $v \in V^-$ for which $\deg{H_p}{v} < p \cdot \deg{H}{v} - \eps p \beta / 2$. By an analogous argument 
to the one above, we have that the expected number of such vertices is at most $\eps^{7}p^{7}\mu(G)$. A vertex can also end up in $V^-$ because it has a 
neighbor in $V^+$ in $H$. But each vertex in $H$ has degree at most $\beta$ so we have $$\expect{|V^-|} \leq 
\eps^{7}p^{7}\mu(G) + \beta \expect{|V^+|} \leq  \eps^{4}p^{4}\mu(G),$$
where the last inequality again uses that $K > C$.
\end{proof}

\begin{remark}\label{rem:independence}
	Interestingly, our result in Theorem~\ref{thm:stochastic} continues to hold as it is even when the edges sampled in realizations of $G_p$ are only $\Theta(1/p)$-wise independent, by simply using a Chernoff bound for bounded-independence random variables (see, e.g.~\cite{SchmidtSS95}) in the proof of Claim~\ref{claim:stochastic}. Allowing correlation in the process of edge sampling is highly relevant to the main application of this problem to the kidney exchange setting (see~\cite{BlumDHPSS15}). To our knowledge, 
	our algorithm is the first to work with such a little amount of independence between the edges in realizations. 
\end{remark}

%% file: stochastic-matching.tex
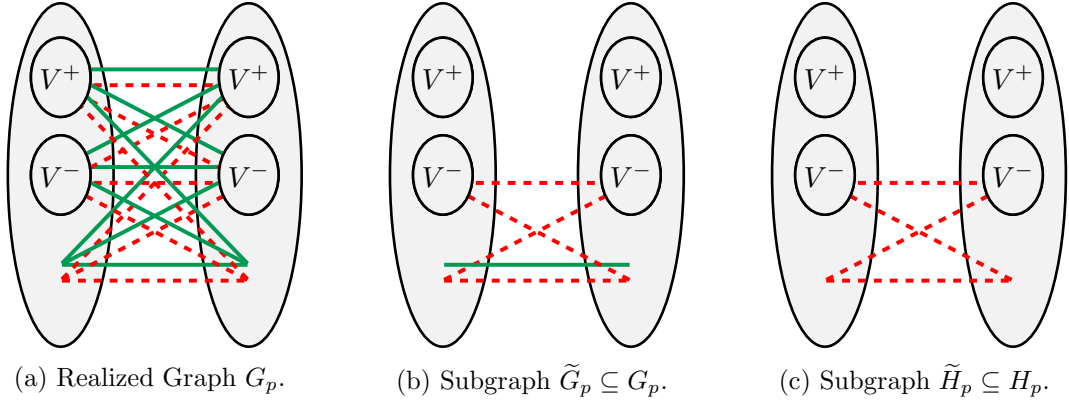
\begin{figure}[t]
    \centering
    \subcaptionbox{Realized Graph $G_p$.}[0.3\textwidth]{
    \begin{tikzpicture}[ auto ,node distance =1cm and 2cm , on grid , semithick , state/.style ={ circle ,top color =white , bottom color = white , draw, black , text=black}, every node/.style={inner sep=0,outer sep=0}]

\node[ellipse, black, draw, fill=black!5, minimum height=130pt, minimum width = 40pt, line width=1pt]  (L) {};
\node[ellipse, black, draw, fill=black!5, minimum height=30pt, minimum width = 20pt, line width=1pt]  (Lp) [above=1.3cm of L]{$V^+$};
\node[ellipse, black, draw, fill=black!5, minimum height=30pt, minimum width = 20pt, line width=1pt]  (Lm) [below=0cm of L]{$V^-$};
\node  (Lx) [below=1.3cm of L]{};

\node[ellipse, black, draw, fill=black!5, minimum height=130pt, minimum width = 40pt, line width=1pt]   (R) [right = 2.5cm of L]{};
\node[ellipse, black, draw, fill=black!5, minimum height=30pt, minimum width = 20pt, line width=1pt]  (Rp) [above=1.3cm of R]{$V^+$};
\node[ellipse, black, draw, fill=black!5, minimum height=30pt, minimum width = 20pt, line width=1pt]  (Rm) [below=0cm of R]{$V^-$};
\node  (Rx) [below=1.3cm of R]{};

\node (lpt)[above right=3pt and 0pt of Lp]{};
\node (lpb)[below right=3pt and 0pt of Lp]{};

\node (rpt)[above left=3pt and 0pt of Rp]{};
\node (rpb)[below left=3pt and 0pt of Rp]{};

\node (lmt)[above right=3pt and 0pt of Lm]{};
\node (lmb)[below right=3pt and 0pt of Lm]{};

\node (rmt)[above left=3pt and 0pt of Rm]{};
\node (rmb)[below left=3pt and 0pt of Rm]{};

\node (lxt)[above right=3pt and 0pt of Lx]{};
\node (lxb)[below right=3pt and 0pt of Lx]{};

\node (rxt)[above left=3pt and 0pt of Rx]{};
\node (rxb)[below left=3pt and 0pt of Rx]{};

\draw[line width=1.5pt, ForestGreen] (lpt) to (rpt);
\draw[line width=1.5pt, red, dashed] (lpb) to (rpb);

\draw[line width=1.5pt, ForestGreen] (lmt) to (rmt);
\draw[line width=1.5pt, red, dashed] (lmb) to (rmb);

\draw[line width=1.5pt, ForestGreen] (lpt) to (rmt);
\draw[line width=1.5pt, red, dashed] (lpb) to (rmb);

\draw[line width=1.5pt, ForestGreen] (lmt) to (rpt);
\draw[line width=1.5pt, red, dashed] (lmb) to (rpb);

\draw[line width=1.5pt, ForestGreen] (lxt) to (rxt);
\draw[line width=1.5pt, red, dashed] (lxb) to (rxb);

\draw[line width=1.5pt, ForestGreen] (lxt) to (rpt);
\draw[line width=1.5pt, red, dashed] (lxb) to (rpb);

\draw[line width=1.5pt, ForestGreen] (lxt) to (rmt);
\draw[line width=1.5pt, red, dashed] (lxb) to (rmb);

\draw[line width=1.5pt, ForestGreen] (rxt) to (lpt);
\draw[line width=1.5pt, red, dashed] (rxb) to (lpb);

\draw[line width=1.5pt, ForestGreen] (rxt) to (lmt);
\draw[line width=1.5pt, red, dashed] (rxb) to (lmb);

\node[ellipse, black, draw, fill=black!5, minimum height=30pt, minimum width = 20pt, line width=1pt]  (Lp) [above=1.3cm of L]{$V^+$};
\node[ellipse, black, draw, fill=black!5, minimum height=30pt, minimum width = 20pt, line width=1pt]  (Lm) [below=0cm of L]{$V^-$};

\node[ellipse, black, draw, fill=black!5, minimum height=30pt, minimum width = 20pt, line width=1pt]  (Rp) [above=1.3cm of R]{$V^+$};
\node[ellipse, black, draw, fill=black!5, minimum height=30pt, minimum width = 20pt, line width=1pt]  (Rm) [below=0cm of R]{$V^-$};

\end{tikzpicture}
    
    }
    \subcaptionbox{Subgraph $\tG_p \subseteq G_p$.}[0.3\textwidth]{
     \begin{tikzpicture}[ auto ,node distance =1cm and 2cm , on grid , semithick , state/.style ={ circle ,top color =white , bottom color = white , draw, black , text=black}, every node/.style={inner sep=0,outer sep=0}]

\node[ellipse, black, draw, fill=black!5, minimum height=130pt, minimum width = 40pt, line width=1pt]  (L) {};
\node[ellipse, black, draw, fill=black!5, minimum height=30pt, minimum width = 20pt, line width=1pt]  (Lp) [above=1.3cm of L]{$V^+$};
\node[ellipse, black, draw, fill=black!5, minimum height=30pt, minimum width = 20pt, line width=1pt]  (Lm) [below=0cm of L]{$V^-$};
\node  (Lx) [below=1.3cm of L]{};

\node[ellipse, black, draw, fill=black!5, minimum height=130pt, minimum width = 40pt, line width=1pt]   (R) [right = 2.5cm of L]{};
\node[ellipse, black, draw, fill=black!5, minimum height=30pt, minimum width = 20pt, line width=1pt]  (Rp) [above=1.3cm of R]{$V^+$};
\node[ellipse, black, draw, fill=black!5, minimum height=30pt, minimum width = 20pt, line width=1pt]  (Rm) [below=0cm of R]{$V^-$};
\node  (Rx) [below=1.3cm of R]{};

\node (lpt)[above right=3pt and 0pt of Lp]{};
\node (lpb)[below right=3pt and 0pt of Lp]{};

\node (rpt)[above left=3pt and 0pt of Rp]{};
\node (rpb)[below left=3pt and 0pt of Rp]{};

\node (lmt)[above right=3pt and 0pt of Lm]{};
\node (lmb)[below right=3pt and 0pt of Lm]{};

\node (rmt)[above left=3pt and 0pt of Rm]{};
\node (rmb)[below left=3pt and 0pt of Rm]{};

\node (lxt)[above right=3pt and 0pt of Lx]{};
\node (lxb)[below right=3pt and 0pt of Lx]{};

\node (rxt)[above left=3pt and 0pt of Rx]{};
\node (rxb)[below left=3pt and 0pt of Rx]{};

\draw[line width=1.5pt, red, dashed] (lmb) to (rmb);

\draw[line width=1.5pt, red, dashed] (lxb) to (rxb);

\draw[line width=1.5pt, red, dashed] (lxb) to (rmb);

\draw[line width=1.5pt, ForestGreen] (lxt) to (rxt);

\draw[line width=1.5pt, red, dashed] (rxb) to (lmb);

\node[ellipse, black, draw, fill=black!5, minimum height=30pt, minimum width = 20pt, line width=1pt]  (Lp) [above=1.3cm of L]{$V^+$};
\node[ellipse, black, draw, fill=black!5, minimum height=30pt, minimum width = 20pt, line width=1pt]  (Lm) [below=0cm of L]{$V^-$};

\node[ellipse, black, draw, fill=black!5, minimum height=30pt, minimum width = 20pt, line width=1pt]  (Rp) [above=1.3cm of R]{$V^+$};
\node[ellipse, black, draw, fill=black!5, minimum height=30pt, minimum width = 20pt, line width=1pt]  (Rm) [below=0cm of R]{$V^-$};

\end{tikzpicture}
    }
    \subcaptionbox{Subgraph $\tH_p \subseteq H_p$.}[0.3\textwidth]{
     \begin{tikzpicture}[ auto ,node distance =1cm and 2cm , on grid , semithick , state/.style ={ circle ,top color =white , bottom color = white , draw, black , text=black}, every node/.style={inner sep=0,outer sep=0}]

\node[ellipse, black, draw, fill=black!5, minimum height=130pt, minimum width = 40pt, line width=1pt]  (L) {};
\node[ellipse, black, draw, fill=black!5, minimum height=30pt, minimum width = 20pt, line width=1pt]  (Lp) [above=1.3cm of L]{$V^+$};
\node[ellipse, black, draw, fill=black!5, minimum height=30pt, minimum width = 20pt, line width=1pt]  (Lm) [below=0cm of L]{$V^-$};
\node  (Lx) [below=1.3cm of L]{};

\node[ellipse, black, draw, fill=black!5, minimum height=130pt, minimum width = 40pt, line width=1pt]   (R) [right = 2.5cm of L]{};
\node[ellipse, black, draw, fill=black!5, minimum height=30pt, minimum width = 20pt, line width=1pt]  (Rp) [above=1.3cm of R]{$V^+$};
\node[ellipse, black, draw, fill=black!5, minimum height=30pt, minimum width = 20pt, line width=1pt]  (Rm) [below=0cm of R]{$V^-$};
\node  (Rx) [below=1.3cm of R]{};

\node (lpt)[above right=3pt and 0pt of Lp]{};
\node (lpb)[below right=3pt and 0pt of Lp]{};

\node (rpt)[above left=3pt and 0pt of Rp]{};
\node (rpb)[below left=3pt and 0pt of Rp]{};

\node (lmt)[above right=3pt and 0pt of Lm]{};
\node (lmb)[below right=3pt and 0pt of Lm]{};

\node (rmt)[above left=3pt and 0pt of Rm]{};
\node (rmb)[below left=3pt and 0pt of Rm]{};

\node (lxt)[above right=3pt and 0pt of Lx]{};
\node (lxb)[below right=3pt and 0pt of Lx]{};

\node (rxt)[above left=3pt and 0pt of Rx]{};
\node (rxb)[below left=3pt and 0pt of Rx]{};

\draw[line width=1.5pt, red, dashed] (lmb) to (rmb);

\draw[line width=1.5pt, red, dashed] (lxb) to (rxb);

\draw[line width=1.5pt, red, dashed] (lxb) to (rmb);

\draw[line width=1.5pt, red, dashed] (rxb) to (lmb);

\node[ellipse, black, draw, fill=black!5, minimum height=30pt, minimum width = 20pt, line width=1pt]  (Lp) [above=1.3cm of L]{$V^+$};
\node[ellipse, black, draw, fill=black!5, minimum height=30pt, minimum width = 20pt, line width=1pt]  (Lm) [below=0cm of L]{$V^-$};

\node[ellipse, black, draw, fill=black!5, minimum height=30pt, minimum width = 20pt, line width=1pt]  (Rp) [above=1.3cm of R]{$V^+$};
\node[ellipse, black, draw, fill=black!5, minimum height=30pt, minimum width = 20pt, line width=1pt]  (Rm) [below=0cm of R]{$V^-$};

\end{tikzpicture}
    }
    \caption{Illustration of the sets $V^+,V^-$ and the subgraphs $\tG_p$ and $\tH_p$ in the proof of Lemma~\ref{lem:stochastic} on a bipartite graph $G$. Here, (green) solid lines
    denote the edges of $G_p$ that appear in each subgraph and (red) dashed lines denote the edges of $H_p$. 
    }
    \label{fig:stochastic-matching}
\end{figure}

%% file: fault-tolerant.tex
\newcommand{\Fstar}{\ensuremath{F^{\star}}}

\newcommand{\mum}{\mu_{\textnormal{min}}}

\section{A Fault-Tolerant Subgraph for Matching}\label{sec:fault-tolerant}

In the fault-tolerant matching problem, we are given a graph $G(V,E)$ and an integer $f \geq 1$, and our goal is to compute a subgraph $H$ of $G$, named an $f$-tolerant subgraph, such that for any subset $F \subseteq E$ of size $f$,  
$H \setminus F$ contains an approximate maximum matching of $G \setminus F$. We show that, 

\begin{theorem}[Formalizing Result~\ref{res:fault-tolerant}]\label{thm:fault-tolerant}
	There exists a deterministic poly-time algorithm that given any $\eps > 0$ and integer $f \geq 1$, computes a $(3/2+\eps)$-approximate $f$-tolerant subgraph $H$ of any given graph $G$ with $O(\eps^{-2} \cdot \paren{n\log{(1/\eps)}+ f})$ edges. 
\end{theorem}
\vspace{-0.25cm}
The algorithm in Theorem~\ref{thm:fault-tolerant} simply computes an EDCS of the input graph as follows: 

\begin{tbox}

\underline{An algorithm for the fault-tolerant matching problem.} 

\begin{enumerate}
\item Define $\mum := \min_{F} \Paren{\mu(G \setminus F)}$, where $F$ is taken over all subsets of $E$ with size $f$. 

\item Output $H:=\EDCS{G,\beta,\beta-1}$ for $\beta := \frac{C \cdot f}{\eps^2 \cdot \mum} + \frac{C \cdot \log{(1/\eps)}}{\eps^2}$ for a large constant $C > 0$.
\end{enumerate}

\end{tbox}
By Proposition~\ref{prop:edcs-exists}, the EDCS $H$ in the above algorithm always exists and can be found in polynomial time. 
The above algorithm as stated however is not a polynomial time algorithm because it is not clear how to compute the quantity $\mum$. Nevertheless, for simplicity, we work with the above algorithm throughout this section, and at the end show how to fix this problem and obtain a poly-time algorithm. We start by proving that the subgraph $H$ only has $O(f + n)$ edges. 

\begin{lemma}\label{lem:ft-size}
	The total number of edges in $H$ is $O(\frac{f}{\eps^2} + n \cdot \frac{\log{(1/\eps)}}{\eps^2})$. 
\end{lemma}
\begin{proof}
	Let $\Fstar$ be a subset of $E$ with size $f$ such that $\mum = \mu(G \setminus \Fstar)$. Let $\Mstar$ be a maximum matching of size $\mum$ in $G \setminus \Fstar$. Note that $V(\Mstar)$ is a vertex cover for 
	$G \setminus \Fstar$. This means that all edges in $G$ except for $f$ of them are incident on $V(\Mstar)$. As no vertex in the EDCS $H$ can have degree more than $\beta$ by $\propone$ of EDCS, the degree
	of vertices in $V(\Mstar)$ in $E \setminus \Fstar$ is at most $\beta$. This implies that: 
	\begin{align*}
		\card{E_H} &\leq \card{V(\Mstar)} \cdot \beta + \card{\Fstar} \leq 2\mum \cdot \paren{\frac{C \cdot f}{\eps^2 \cdot \mum} + \frac{C \cdot \log{(1/\eps)}}{\eps^2}} + f = O(\frac{f}{\eps^2} + n \cdot \frac{\log{(1/\eps)}}{\eps^2}),
	\end{align*}
	finalizing the proof.
\end{proof}

We now prove the correctness of the algorithm in the following lemma. 

\begin{lemma}\label{lem:fault-tolerant}
	Fix any subset $F \subseteq E$ of size $f$ and define $G_F := G \setminus F$ and $H_F := H \setminus F$. Then, $\MM{G_F} \leq (3/2 + O(\eps)) \cdot \MM{H_F}$.
\end{lemma}

We first need some definitions. We say that a vertex $v \in V$ is \emph{bad} iff $\deg{H_F}{v} < \deg{H}{v} - \eps\beta$, i.e., at least $\eps\beta$ edges incident on $v$ (in $H$) are deleted by $F$. 
We use $B_F$ to denote the set of bad vertices with respect to $F$, and bound $\card{B_F}$ in the following claim. 

\begin{claim}\label{clm:bad-bound}
	Number of bad vertices in $H_F$ is at most $\card{B_F} \leq \eps \cdot \mu(G_F)$.
\end{claim}
\begin{proof}
	Any deleted edge can decrease the degrees of exactly two vertices. Any vertex becomes bad iff at least $\eps\beta$ edges incident on it from $H_F$ are removed. As such, 
	$	\card{B_F} \leq \frac{2f}{\eps \cdot \beta} \leq \frac{2f \cdot \eps^2 \cdot \mum}{\eps \cdot C \cdot f} \leq \eps \cdot \mu(G_F),$
	for sufficiently large $C > 0$, and since $\mu(G_F) \geq \mum$ by definition of $\mum$. \Qed{Claim~\ref{clm:bad-bound}}
	
\end{proof}

\begin{proof}[Proof of Lemma~\ref{lem:fault-tolerant}]
Define a subgraph $\tG_F \subseteq G_F$ as follows: $V(\tG_F) = V(G_F)$ ($=V(G)$) and edges in $\tG_F$ are all edges in $G_F$ except that we remove any edge $(u,v) \in G_F$ such that
$(u,v) \notin H_F$ and either of $u$ or $v$ is a bad vertex. We prove that $\mu(\tG_F)$ is at least $(1-\eps)$ fraction of $\mu(G_F)$, and moreover, $H_F$ is an EDCS of $\tG_F$ with appropriate parameters. 
We can then apply Lemma~\ref{lem:general-edcs-matching} to obtain that $\mu(G_F) \leq (1+2\eps) \mu(\tG_F) \leq (1+\eps)\cdot (3/2+O(\eps)) \mu(H_F)$, finalizing the proof.  

We first prove the bound on $\mu(\tG_F)$. Fix any maximum matching $M$ in $G_F$. It can have at most $\card{B_F}$ edges incident on vertices of $B_F$. Hence, even if we remove all edges incident on $B_F$, the size of this matching would be
at least $\mu(G_F) - \eps \cdot \mu(G_F)$, by the bound of $\card{B_F} \leq \eps \cdot \mu(G_F)$ in Claim~\ref{clm:bad-bound}. However, this matching belongs to $\tG_F$ entirely by the definition of this subgraph, and hence we have, 
$\mu(G_F) \leq (1+2\eps) \mu (\tG_F)$. 

We now prove that $H_F$ is an $\EDCS{\tG_F,\beta,(1-2\eps)\beta-1}$ of $\tG_F$. It suffices to prove the two properties of EDCS for $H_F$ using the fact that $\deg{H_F}{v} \in [\deg{H}{v} - \eps \beta,\deg{H}{v}]$ for vertices in $V \setminus B_F$, and that all edges incident on $B_F$ in $\tG_F$ also belong to $H_F$.  
\begin{itemize}
	\item \propone of EDCS $H_F$ of $\tG_F$: For any edge $(u,v) \in H_F$:
	\begin{align*}
		\deg{H_F}{u} + \deg{H_F}{v} \leq \deg{H}{u} + \deg{H}{v} \leq \beta. \tag{by $\propone$ of EDCS $H$ of $G$}
	\end{align*}
	\item \proptwo of EDCS $H_F$ of $\tG_F$: For any edge $(u,v) \in \tG_F \setminus H_F$ both $u,v \in V \setminus B_F$ and so: 
	\begin{align*}
		\deg{H_F}{u} + \deg{H_F}{v} \geq \deg{H}{u} + \deg{H}{v} - 2\eps\beta \geq (1-2\eps)\beta - 1. \tag{by $\proptwo$ of EDCS $H$ of $G$ as $(u,v)$ is missing from $H$}
	\end{align*}
\end{itemize}
As such, $H_F$ is an $\EDCS{\tG_F,\beta,(1-2\eps)\beta-1}$ of $\tG_F$ and by the lower bound on value of $\beta$ in the algorithm (the second term in definition of $\beta$), we can apply Lemma~\ref{lem:general-edcs-matching},
and obtain that $\MM{\tG_F} \leq (3/2+O(\eps)) \cdot \MM{H_F}$, finalizing the proof. 
\end{proof}

Theorem~\ref{thm:fault-tolerant} now follows from Lemmas~\ref{lem:ft-size} and~\ref{lem:fault-tolerant} by re-parametrizing $\eps$ to a sufficiently smaller constant factor of $\eps$ (by picking the integer $C$ large enough) modulo the fact that the algorithm designed in this section is not a polynomial time algorithm. To make the algorithm polynomial time, we only
need to make a simple modification: instead of finding $\mum$ explicitly, we find the smallest value of $\beta$ (by searching over all $n$ possible choices of $\beta$, or by doing a binary search) such that the EDCS $H$ has 
at least $\frac{2 \cdot C \cdot f}{\eps^2} + \frac{n \cdot C \cdot \log{(1/\eps)}}{\eps^2}$ many edges. By the proof of Lemma~\ref{lem:ft-size}, any EDCS of $G$ can have at most 
$2\mum \cdot \beta + f$ edges. This implies that the chosen $\beta \geq \frac{C \cdot f}{\eps^2 \cdot \mum} + \frac{C \cdot \log{(1/\eps)}}{\eps^2}$ as needed in the algorithm. This concludes the proof, 
as by definition of $\beta$, $H$ has $O(\frac{C \cdot f}{\eps^2} + \frac{n \cdot C \cdot \log{(1/\eps)}}{\eps^2})$ many edges, and hence satisfies the sparsity requirements of Theorem \ref{thm:fault-tolerant}.

%% file: appendix.tex
\clearpage
\section{Missing Details and Proofs}\label{app:appendix}

\subsection{Proof of Proposition~\ref{prop:edcs-exists}}\label{app:edcs-exists}

We give the proof of this proposition following the argument of~\cite{AssadiBBMS17}, which itself was based on~\cite{BernsteinS16}. 

\begin{proof}
	
	We give a polynomial local search algorithm for constructing an EDCS $H$ of the graph $G$ which also implies the existence of $H$. 
	The algorithm is as follows. Start with empty graph $H$. While there exists an edge in $H$ or $G \setminus H$ that violates $\propone$ or $\proptwo$ of EDCS, respectively, {fix} this edge by removing it 
	from $H$ for the former or inserting it to $H$ for the latter. 
	
	We prove that this algorithm terminates after polynomial number of steps which implies both the existence of the EDCS as well as give a polynomial time algorithm for computing it. 
	We define the following potential function $\Phi$ for this task: 
	\begin{align*}
		\Phi_1(H) &:= \paren{\beta - 1/2} \cdot \sum_{u \in V(H)} \deg{H}{u}, \qquad \Phi_2(H) := \sum_{(u,v) \in E(H)} \paren{\deg{H}{u} + \deg{H}{v}}, \\
		\Phi(H) &:= \Phi_1(H) - \Phi_2(H). 
	\end{align*}
	We claim that after fixing each edge in $H$ in the algorithm, $\Phi$ increases by at least $1$. 
	Since max-value of $\Phi$ is $O(n \cdot \beta^2)$, this implies that this procedure terminates in $O(n \cdot \beta^2)$ steps.
	
	Let $(u,v)$ be the fixed edge at this step, $H_1$ be the subgraph before fixing the edge $(u,v)$, and $H_2$ be the resulting subgraph. 
	Suppose first that the edge $(u,v)$ was violating \propone of EDCS.  As the only change is in the degrees of vertices $u$ and $v$,
	$\Phi_1$ decreases by $(2\beta-1)$. On the other hand, $\deg{H_1}{u} + \deg{H_1}{v} \geq \beta+1$ originally (as $(u,v)$ was violating \propone of EDCS) and hence after removing $(u,v)$, $\Phi_2$ also
	decreases by $\beta+1$. Additionally, for each neighbor $w$ of $u$ and $v$ in $H_2$, after removing the edge $(u,v)$, $\deg{H_2}{w}$ decreases by one. As there are at least 
	$\deg{H_2}{u} + \deg{H_2}{v} = \deg{H_1}{u} + \deg{H_1}{v} -2 \geq \beta - 1$ choices for $w$, 
	this means that in total, $\Phi_2$ decreases by at least $(\beta + 1) + (\beta - 1) = 2\beta$. As a result, in this case $\Phi = \Phi_1 - \Phi_2$ increases by at least $1$ after fixing the edge $(u,v)$. 
	
	Now suppose that the edge $(u,v)$ was violating \proptwo of EDCS instead. In this case, degree of vertices $u$ and $v$ both increase by one, hence $\Phi_1$ increases by $2\beta-1$. 
	Additionally, since edge $(u,v)$ was violating $\proptwo$ we have $\deg{H_1}{u} + \deg{H_1}{v} \leq \beta^- - 1$, so the addition of edge $(u,v)$ decreases $\Phi_2$ by at most $\deg{H_2}{u} + \deg{H_2}{v} = \deg{H_1}{u} + \deg{H_1}{v} + 2 \leq \beta^- + 1$.
	Moreover, for each neighbor $w$ of $u$ and $v$, after adding the edge $(u,v)$, $\deg{H_2}{w}$ increases by one and since
	there are at most $\deg{H_1}{u} + \deg{H_1}{v} \leq \beta^- - 1$ choices for $w$, $\Phi_2$ decreases in total by at most $(\beta^- + 1) + (\beta^- - 1) = 2\beta^-$. 
	Since $\beta^-  \leq \beta-1$, we have that $\Phi$ increases by at least 
	$(2\beta-1) - (2\beta^-) \geq 1$ after fixing the edge $(u,v)$, finalizing the proof. 
\end{proof}

\subsection{Optimality of the $(3/2)$-Approximation Ratio in Result~\ref{res:fault-tolerant}}\label{app:optimal-fl}

Our argument is a simple modification of the one in~\cite{GoelKK12} for proving a lower bound on the one-way communication complexity of approximating matching and is provided  for the sake of completeness.
 
Let $G_1(V_1,E_1)$ be a graph on $N$ vertices such that its edges can be partitioned into $t:= N^{\Omega(1/\log\log{N})}$ \emph{induced} matchings $M_1,\ldots,M_t$ of size $(1-\delta)N/4$ for arbitrarily small constant $\delta > 0$. These graphs are referred to as $(r,t)$-\rs graphs~\cite{RuszaS78} ($(r,t)$-RS graphs for short) 
and have been studied extensively in the literature (see~\cite{AssadiKL17,GoelKK12} for more details). In particular, the existence of such graphs with
parameters mentioned above is proven in~\cite{GoelKK12}. 

Let $G(V,E)$ be a graph with $n=2N$ vertices consisting of $G_1(V_1,E_1)$ plus $N$ additional vertices $U$ that are connected via a perfect matching $M_U$ to $V_1$. In the following, we prove 
that any $f$-fault tolerant subgraph $H$ of $G$ that achieves a $(3/2-\eps)$-approximation for some constant $\eps > 0$ when $f = \Theta(n)$ requires $n^{1+\Omega(1/\log\log{n})} = \omega(f)$ edges. 

Suppose towards a contradiction that $H$ contains $o(m)$ edges where $m$ is the number of edges in the graph $G$. As edges in $G_1$ are partitioned into induced matchings $M_1,\ldots,M_t$, 
it means that there exists some induced matching $M_i$ such that only $o(1)$ fraction of its edges belong to $H$. Let the set of deleted edge $F$ be only the set of edges in the perfect matching between $U$ and $V_1$, namely, $M_U$,
which are incident to $V(M_i)$. The number of deleted edges is $O(n)$ and after deletion, $M_U$ has size $N-(1-\delta)N/2 = (1+\delta)N/2$. As such, $\mu(G \setminus F) \geq (1+\delta)N/2 + (1-\delta)N/4 \geq 3N/4$, by picking the remainder
of the matching $M_U$ and the induced matching $M_i$ (which is not incident on remainder of $M_U$ by construction). However, we argue that $\mu(H \setminus F) \leq (1+\delta)N/2 + o(N)$, simply because only $o(N)$ edges of $M_i$ 
belong $H$ and all other matchings are incident to the remaining edges of $M_U$ (we can assume remaining edges of $M_U$ belong to \emph{any} maximum matching of $H \setminus F$ because they are incident on degree one vertices). 
As such, $\mu(H \setminus F) < (2/3+2\delta)\mu(G \setminus F)$. By picking $\delta < \eps/4$, we obtain that $H$ is not a $(3/2-\eps)$-approximate $f$-fault tolerant subgraph of $G$.

\subsection{Other Standard Algorithms for Fault-Tolerant Matching}\label{app:fault-tolerant}

Since the goal in fault-tolerant matching is to prepare for adversarial deletions, the most natural approach seem to be adding many different matchings by a finding maximum matching in $G$, adding it to the subgraph $H$, deleting it from $G$, and repeating until we have $O(f + n)$ edges. A similar approach would be to let $H$ be a maximum $b$-matching, with $b$ set appropriately to end up with $O(f + n)$ edges. 
We show a lower bound of $2$ on the approximation ratio of these approaches. 

Consider the following approach first: find a maximum matching $M$ in $G$, add all the edges of $M$ to the fault-tolerant subgraph $H$, remove all the edges of $M$ from $G$, and repeat until the graph contains $C(f+n)$ edges for some large constant $C$. For $f = n/5$, we present a graph $G$ where this approach yields a graph $H$ where $\mu(H) = \mu(G)/2$. The graph is bipartite and the vertex set is partitioned into $5$ sets $X,Y,Y',Z,Z'$, each of size $n/5$.  There is an edge in $G$ from every vertex in $X$ to every vertex in $Y$ or $Z$, and there are also exactly $n/5$ vertex-disjoint edges from $Y$ to $Y'$, and similarly from $Z$ to $Z'$; those are all the edge of $G$. The fault tolerant algorithm might choose the following subgraph $H$: $H$ contains a perfect matching from $Y$ to $Y'$ and from $Z$ to $Z'$, as well as many edges from $X$ to $Y$, but no edges from $X$ to $Z$. (The algorithm can end up with such an $H$ by first choosing the maximum matching in $G$ that consists of the edges from $Y$ to $Y'$ and from $Z$ to $Z'$; then for all future iterations the maximum matching size is only $|X| = n/5$, so the algorithm might always pick a maximum matching that only contains edges between $X$ and $Y$.) Now consider the set of failures $F$ which consists of the $n/5$ edges from $Z$ to $Z'$. It is clear that $\mu(G \setminus F) = 2n/5$, while $\mu(H \setminus F) = n/5$. Note also that allowing $H$ to contain more than $O(n + f)$ edges would still not allow this approach to break through the $2$-approximation: in this lower-bound instance, even if $H$ was allowed to have up to $n^2/100$ edges, $H$ might still not contain any edges from $X$ to $Z$, and so we would still have  $\mu(H \setminus F) = n/5 =  \mu(G \setminus F)/2$.

The other natural approach is to let $H$ contain the edges of a maximum $b$-matching in $G$, where $b$ is set to a value for which the resulting $b$-matching still contains  $\Theta(f + n)$ edges. The lower-bound graph $G$ is exactly the same as above, though in this case we use $f = 2n/5$. The maximum $b$-matching $H$ might then contain the edges from $Y$ to $Y'$ and $Z$ to $Z'$, a single matching of size $n/5$ from $X$ to $Z$, and then many edges from $X$ to $Y$. It is easy to see that this is a maximum $b$-matching. Now consider the following set $F$ of deletions: $F$ contains all edges from $Z$ to $Z'$, as well as the $n/5$ edges in $H$ from $X$ to $Z$. It is easy to see that we once again have $\mu(H) = n/5$ and $\mu(G) = 2n/5$. Also as above, setting $B$ to be very large and allowing $H$ to have $n^2/100$ edges would still not break through the $2$-approximation.
